\relax
\documentclass[letterpaper]{article} 
\usepackage{aaai19}  
\usepackage{times}  
\usepackage{helvet}  
\usepackage{courier}  
\usepackage{url}  
\usepackage{graphicx}  
\usepackage{amsmath}
\usepackage{amsthm}

\newtheorem{innercustomgeneric}{\customgenericname}
\providecommand{\customgenericname}{}
\newcommand{\newcustomtheorem}[2]{%
  \newenvironment{#1}[1]
  {%
   \renewcommand\customgenericname{#2 }%
   \renewcommand\theinnercustomgeneric{##1 }%
   \innercustomgeneric
  }
  {\endinnercustomgeneric}
}
\newcustomtheorem{customthm}{Theorem}
\newtheorem{lemma}{Lemma}

\usepackage{microtype}
\usepackage{graphicx}
\usepackage{booktabs} 
\usepackage{amssymb}
\usepackage{paralist}
\usepackage{relsize}
\usepackage{rotating}
\usepackage[export]{adjustbox}
\usepackage{multirow}
\usepackage{caption} 
\usepackage{subcaption}

\usepackage{amsfonts}

\begin{document}

\title{\textbf{A Correlation Maximization Approach for Cross Domain Co-Embeddings}}
\author{
Dan Shiebler\\
Twitter Cortex\\
dshiebler@twitter.com
}
 \date{} 
\maketitle

\newcommand{\simUBij}{sim(e_{U_{i_B}}, e_{B_j})}
\newcommand{\simUBijBig}{sim\big(e_{U_{i_B}}, e_{B_j}\big)}

\newcommand{\Lstar}{L^{\ast}}

\newcommand{\Li}{L_{i}}
\newcommand{\Lihat}{\widehat{L_{i}}}

\newcommand{\Lcorr}{L^{c}}
\newcommand{\Licorr}{L_i^{c}}
\newcommand{\LcorrS}{L^{c_S}}
\newcommand{\LicorrS}{L_i^{c_S}}

\newcommand{\sumPMinusISquare}{\sum\limits_{}^j \left(P_{ij} - Y_{ij}\right)^{2}}

\newcommand{\sumPMinusISquareNormalized}{\sum\limits_{}^j \left(P_{ij} - \widehat{Y_{ij}}\right)^{2}}

\newcommand{\sumSQRTPMinusISquareNormalized}{\sum\limits_{}^j \sqrt{\left(P_{ij} - \widehat{Y_{ij}}\right)^{2}}}

\newcommand{\SumIhatminusIbar}{\sum\limits_{}^j \bigg(\widehat{Y}_{ij} - \overline{\widehat{Y}_{i}}\bigg)^2}

\newcommand{\SumPhatminusPbar}{\sum\limits_{}^j \bigg(\widehat{P}_{ij} - \overline{\widehat{P}_{i}}\bigg)^2}

\newcommand{\SumIhatminusIbarS}{\sum\limits_{}^{j\in S} \bigg(\widehat{Y}_{ij} - \overline{\widehat{Y}_{iS}}\bigg)^2}

\newcommand{\SumPhatminusPbarS}{\sum\limits_{}^{j\in S} \bigg(\widehat{P}_{ij} - \overline{\widehat{P}_{iS}}\bigg)^2}

\newcommand{\SumPhatminusPbarIhatminusIbar}
{
\sum\limits_{}^{j\in S} \bigg(\widehat{P}_{ij} - \overline{\widehat{P}_{i}}\bigg) \bigg(\widehat{Y}_{ij} - \overline{\widehat{Y}_{i}}\bigg)
}

\newcommand{\SumIminusIbar}{\sum\limits_{}^j \bigg(Y_{ij} - \overline{Y_{i}}\bigg)^2}

\newcommand{\SumPminusPbar}{\sum\limits_{}^j \bigg(P_{ij} - \overline{P_{i}}\bigg)^2}


\newcommand{\SumPminusPbarIminusIbar}
{
\sum\limits_{}^{j} \bigg(P_{ij} - \overline{P_{i}}\bigg) \bigg(Y_{ij} - \overline{Y_{i}}\bigg)
}

\newcommand{\perUserLoss}
{
\left(1 - \frac {
\SumPminusPbarIminusIbar
}
{
  \sqrt{
    \SumPminusPbar
    \:
    \SumIminusIbar
  }
}
\right)
}

\newcommand{\gradientPearsonLoss}{
-\left(
\frac {
\bigg(Y_{ij} - \overline{Y_{i}}\bigg) 
-
 	\frac {\SumPminusPbarIminusIbar}
    	  {\SumPminusPbar}
 \bigg(P_{ij} - \overline{P_{i}}\bigg) 
} {
  \sqrt{
    \SumPminusPbar
    \:
    \SumIminusIbar
  }
}
\right)
}

\newcommand{\bigUserSum}{
\frac{1}{N_U}\sum\limits_{}^{i} 
}
\newcommand{\bigUserSumS}{
\frac{1}{N_{U_S}}\sum\limits_{}^{i\in S_U} 
}

\newcommand{\negativeBigUserSumS}{
\frac{-1}{N_{U_S}}\sum\limits_{}^{i\in S_U} 
}

\newcommand{\smallUserSumS}{
\frac{1}{N_{U_S}}\sum\limits_{}^{i\in S_U}
}

\newcommand{\SumIminusIbarS}{\sum\limits_{}^{j\in S_I} (Y_{ij} - \overline{Y_{iS_I}})^2}

\newcommand{\SumPminusPbarS}{\sum\limits_{}^{j\in S_I} (P_{ij} - \overline{P_{iS_I}})^2}

\newcommand{\SumPminusPbarIminusIbarS}
{
\sum\limits_{}^{j\in S_I} (P_{ij} - \overline{P_{iS_I}}) (Y_{ij} - \overline{Y_{iS_I}})
}

\newcommand{\perUserLossS}
{
\left(1 - \frac {
\SumPminusPbarIminusIbarS
}
{
  \sqrt{
    \SumPminusPbarS
    \:
    \SumIminusIbarS
  }
}
\right)
}

\newcommand{\gradientPearsonLossS}{
\left(
\frac {
(Y_{ij} - \overline{Y_{iS}}) 
-
 	\frac {\SumPminusPbarIminusIbarS}
    	  {\SumPminusPbarS}
 (P_{ij} - \overline{P_{iS}}) 
} {
  \sqrt{
    \SumPminusPbarS
    \:
    \SumIminusIbarS
  }
}
\right)
}

\begin{abstract}
Although modern recommendation systems can exploit the structure in users' item feedback, most are powerless in the face of new users who provide no structure for them to exploit. In this paper we introduce ImplicitCE, an algorithm for recommending items to new users during their sign-up flow. ImplicitCE works by transforming users' implicit feedback towards auxiliary domain items into an embedding in the target domain item embedding space. ImplicitCE learns these embedding spaces and transformation function in an end-to-end fashion and can co-embed users and items with any differentiable similarity function. 

To train ImplicitCE we explore methods for maximizing the correlations between model predictions and users' affinities and introduce Sample Correlation Update, a novel and extremely simple training strategy. Finally, we show that ImplicitCE trained with Sample Correlation Update outperforms a variety of state of the art algorithms and loss functions on both a large scale Twitter dataset and the DBLP dataset.
\end{abstract}

\section{Introduction}
In today's world of limitless entertainment, the competition for attention is fiercer than ever. When users open a site or app, they expect to see something that they like immediately. In response to this competition, researchers have developed powerful collaborative filtering algorithms that predict which new items users will like based on the structure in the user-item affinity graph.

Popular approaches have historically included neighborhood approaches which predict user affinity by explicitly grouping users and items \cite{neighborhood1} \cite{neighborhood2} and model based algorithms such as matrix factorization \cite{svdpp} \cite{implicit} \cite{pmf}. Recently, researchers have shown success with methods that exploit nonlinear user-item relationships such as autoencoders \cite{cdlcf} \cite{autoencodercf}, RBMs \cite{rbmcf} and supervised deep neural networks \cite{ncf}. Many of these algorithms frame recommendation as a ``reconstruction'' problem, where the objective is to ``fill in the gaps'' in incomplete user-item affinity information \cite{svdpp} \cite{implicit} \cite{pmf} \cite{autoencodercf}.. 

One of the largest draws of matrix factorization and certain deep collaborative filtering methods like \cite{ncf} is that these methods yield low dimensional user and item embeddings. In large multi-component systems these embeddings can be used as information dense inputs to other machine learning models. 
However, the user and item embeddings matrix factorization generates have another desirable property: they are dot product co-embeddings. That is, we can estimate user-item affinity with only an embedding dot product, instead of an expensive neural network evaluation.

For most collaborative filtering algorithms, it is difficult to generate embeddings or make recommendations for new users. One approach to this ``user cold start'' problem is to utilize users' actions in an auxiliary domain in order to inform recommendation in the target domain. In this paper we:

\begin{compactitem} \itemsep0em
\item Introduce ImplicitCE, an algorithm that transforms user's implicit feedback towards auxiliary domain items into a co-embedding in a target domain item embedding space, and illustrate how we can use ImplicitCE to recommend target domain items to new users.
\end{compactitem} \itemsep0em

\begin{compactitem} \itemsep0em
\item Demonstrate that directly maximizing the correlations between model predictions and each user's affinities can yield better performance on the auxiliary domain implicit feedback recommendation task than minimizing a mean square error or ranking loss. 

\end{compactitem} \itemsep0em
\begin{compactitem} \itemsep0em
\item Introduce Sample Correlation Update, a novel, efficient, and incredibly simple method for maximizing these correlations.
\end{compactitem} \itemsep0em

\begin{compactitem} \itemsep0em
\item Evaluate ImplicitCE and Sample Correlation Update on both a large scale Twitter dataset and the public DBLP citation dataset and show that they outperform baseline methods and loss functions on a variety of performance metrics.
\end{compactitem} \itemsep0em

\section{Related Work}

Many cross domain models rely on transfer learning at the latent factor or cluster level to join domains. Some like \cite{socialtensor} \cite{tagcdcf} \cite{transfergenerative} use user-provided cross domain tags such as genre to add model structure like additional matrix dimensions or factorization objective constraints. Others like the collective matrix factorization model in \cite{collectivematrixfactorization} work to exploit the structure that underlies user-item affinity matrices in multiple domains.

However, these approaches tend to provide little additional value in the true cold start situation, where the user has had no interactions with items in the target domain. The simplest strategy to handle this problem is to concatenate the user-item interaction profiles in each domain into a joint domain interaction profile and then perform traditional collaborative filtering techniques like in \cite{userprofilesize}. Another common strategy is to develop a mapping from user-item interactions in the source domain to interactions in the target domain. For example, if we use CCA to compute the correlation matrix $P$ and the canonical component matrices $W_x$ and $W_y$ , we can use these matrices to project the source domain user-item interaction matrix $X$ to an estimate of the target domain user-item interaction matrix $Y$ with $\widehat{Y} = XW_xPW_y^{T}$ \cite{ccacold}.
Recently, neural methods for learning this mapping have grown in popularity. In \cite{starspace}, the authors describe an algorithm for co-embedding entities based on positive entity pairs and demonstrate how it can be used to generate user-item recommendations from binary feedback data. In \cite{microsoftdeepcrossdomain} the authors describe a multi-view deep neural network model with feature hashing that embeds users and items from multiple domains in the same space. In \cite{embeddingtoembedding} the authors introduce a method for mapping between user embedding spaces.

To our knowledge, methods to directly maximize the sum of the Pearson correlations between users' predicted and demonstrated item affinities have not been previously studied. However, many neighborhood-based recommendation systems use Pearson correlation as an item affinity vector similarity metric \cite{neighborhoodcorrelation1} \cite{neighborhoodcorrelation2}. In addition, correlation has been used as a loss function in methods where the exact magnitude of a continuous output is less important than its relative value, such as Cascade Correlation Networks \cite{cascadecorr}.

\section{ImplicitCE and Sample Correlation Update}

We are considering the problem of recommending to a new user a set of items that belong to some target domain under the constraint that the user has interacted with items that belong to some auxiliary domain but has not interacted with any items in the target domain. For both domains, we use the strength of a user's interaction with an item as a measure of that user's ``affinity'' for that item. 

For example, we use the number of times that a user visits a news website as a measure of that user's affinity for that news website.

Note that this is a form of graded implicit feedback data, and that we consider the absence of interaction between a user and an item to be indicative of low affinity between that user and that item. Since we are using only auxiliary domain data to predict target domain affinity, this does not directly damage our model's performance on out-of-bag user-item pairs.

We propose ImplicitCE, an end-to-end framework for generating co-embeddings of users and target domain items. In this framework, a user's target domain embedding is a function of their auxiliary domain item affinities, and the predicted affinity between a user and a target domain item is determined by the similarity between their target domain embedding and that item's embedding.

\subsection{Building and Using ImplicitCE}
\begin{figure}
\includegraphics[height=4.5in,  width=4in, keepaspectratio]{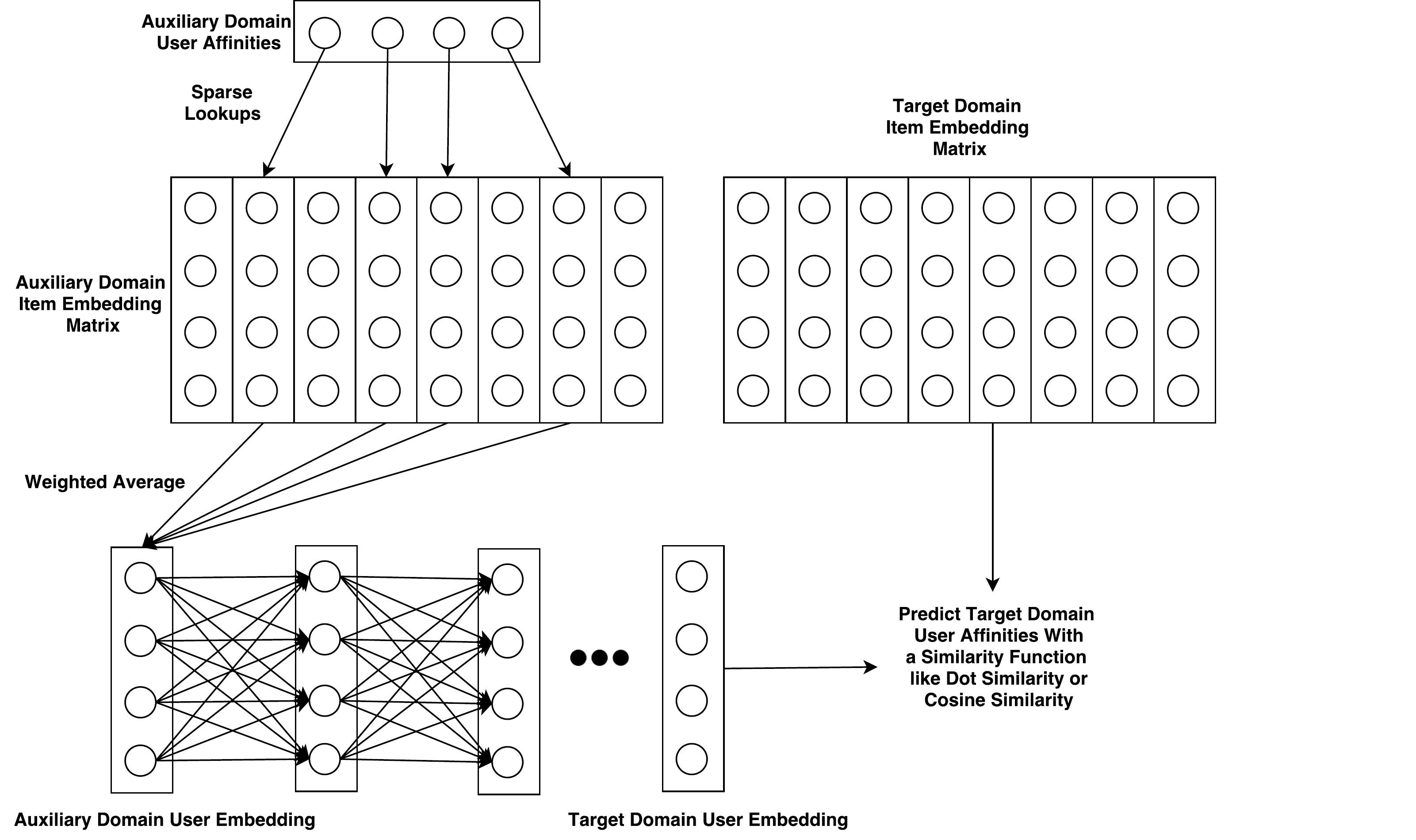}
\caption{ImplicitCE. A user's auxiliary domain embedding is the affinity-weighted average of the auxiliary domain item embeddings. A neural network maps this embedding to a target domain user embedding that we can compare to target domain items' embeddings with a similarity function.}
\end{figure}

ImplicitCE consists of three components that are learned simultaneously: the embedding map $e_A$ which assigns embeddings to each item $a_i$ in the set of auxiliary domain items $A$, the embedding map $e_B$ which assigns embeddings to each item $b_j$ in the set of target domain items $B$, and the transformation function $f(e_{U_A};\theta)$ which transforms user $u$'s auxiliary domain embedding $e_{U_A}$ into a target domain embedding $e_{U_B}$. 

ImplicitCE computes users' auxiliary domain embeddings with an affinity-weighted linear combination of auxiliary domain item embeddings. That is, if $e_{A_i}$ is the embedding of item $a_i$ and $k_{ai}$ is $u$'s affinity for $a_i$, then $u$'s auxiliary domain embedding $e_{U_A}$ is $k_{a1}*e_{A_1} + k_{a2}*e_{A_2} + ...$ and $u$'s target domain embedding $e_{U_B}$ is $f(e_{U_A})$. We can then assess the strength of a user's affinity for some item $b_j$ in $B$ as $sim(e_{U_B}, e_{B_j})$, where $sim(u, v)$ is a function such as dot product. If we are not planning on utilizing the embeddings with an approximate nearest neighbor system, we can also add per-user or per-item bias terms.

There are several significant benefits to this framework. First, ImplicitCE can immediately generate target domain recommendations for new users who were not present at model training time and have had no interactions with any items in the target domain. Furthermore, ImplicitCE does not require any content information about the items in the auxiliary or target domains. Moreover, ImplicitCE generates user embeddings in the target domain rather than directly predicting affinity. This is more efficient than a method that requires a neural network evaluation for each user-item pair like \cite{ncf}. In addition, since ImplicitCE can construct these user-target co-embeddings based on any differentiable embedding similarity function,
including metrics like cosine or euclidian similarity whose positive complements are true distance metrics, we can use ImplicitCE co-embeddings with approximate nearest neighbor algorithms like LSH to match items to users extremely efficiently. 

Furthermore, since ImplicitCE learns the auxiliary and target embedding spaces along with the function to transform between them, it can construct the embeddings to exploit the joint distribution of $P(a_1, a_2...,b_1, b_2,...)$ rather than just the marginal distributions $P(a_1, a_2, ...)$ and $P(b_1, b_2, ...)$. To demonstrate this difference let's consider an extreme example. Say there are two items $a_i, a_j$ in $A$ such that users' affinity for $a_i$ is highly correlated with their affinity for $a_j$. If we use a latent factor model like SVD, then affinities for these two items are likely to be collapsed into a single dimension, and a user's relative affinities for $a_i$ and $a_j$ will have a much less significant impact on $e_{U_A}$ than the average of that user's affinities for $a_i$ and $a_j$. However, if it is the case that the difference between the degrees of a user's interaction with $a_i$ and $a_j$ is the most important signal for predicting a user's interaction with items in $B$, this will be difficult for a model that is trained on the SVD latent factor representations to learn. 
\subsection{Training ImplicitCE} \label{section:ImplicitCESection}
The objective of ImplicitCE is to generate target domain user embeddings such that $\simUBij$ is correlated with the affinity between user $u_i$ and item $b_j$. A standard way to do this is to use a variant of the technique from \cite{pmf} and model the conditional distribution over the user-item interaction counts with $\mathcal{N}\big(Y_{ij} | \simUBijBig, \sigma^2\big)$ where $Y_{ij}$ is the number of interactions between user $u_i$ and item $b_j$ and $\mathcal{N}(x | \mu, \sigma^2)$ is the probability density function of the Gaussian distribution with mean $\mu$ and variance $\sigma^2$. Then the task of maximizing the likelihood of data over this distribution is equivalent to minimizing the square error loss: $\sum\limits_{}^i \sum\limits_{}^j \big(\simUBijBig - Y_{ij}\big)^{2}$.

However, the assumptions of the above model don't generally hold, since $Var(Y_{ij})$ is not constant for all $i$. Users with more target domain interactions can dominate the loss. Moreover, our goal is for the user-item embedding similarities to be correlated with user-item affinity. It is unimportant whether their magnitudes are close to the exact numbers of interactions.

An alternative approach is to frame the problem as a personalized ranking problem and aim to make the ranking of items that the model produces for each user be as close as possible to the actual rank of the items by user interaction. A popular way to do this is to use a pairwise ranking objective that casts the ranking problem as a classification problem. At each step, we sample a user $u_i$ and a pair of items $b_{j_1}, b_{j_2}$ such that $u_i$ has a greater affinity for $b_{j_2}$ than $b_{j_1}$. The loss is some function of $sim\big(e_{U_{i_B}}, e_{B_{j_1}}\big)$ and $sim\big(e_{U_{i_B}}, e_{B_{j_2}}\big)$. For example, in BPR the loss is: $\sum\limits_{}^{i, b_{j_1}, b_{j_2}} -\ln S \big(sim\big(e_{U_{i_B}}, e_{B_{j_2}}\big) - sim\big(e_{U_{i_B}}, e_{B_{j_1}}\big)\big)$ where $S$ is the sigmoid function.

One aspect of ranking objectives is that they do not attempt to capture the shape of a user's affinity function. Consider a user who has several distinct item affinity groups, such that within each group the user likes all items about the same. Then any ranking that correctly groups the items will be mostly true to that user's preferences. However, it is possible for the ranking loss to provide the same or an even greater penalty for improperly ordering items within groups than across groups. That is, it is possible for the predicted affinity to be highly correlated with the number of interactions and for the ranking loss to be large, and it is possible for the predicted affinity to be largely uncorrelated with the number of interactions but for the ranking loss to be small (See Appendix B for an example).

\subsubsection{User-Normalized MSE and Per-User Correlation Loss}

We can avoid the problems of both of the above approaches by adopting a modified version of the mean square error loss. Lets consider some user $u_i$, the vector $Y_{i}$ of $u_i$'s target domain interactions, and the vector $P_{i}$ of the model's predictions of $u_i$'s target domain interactions. That is, $P_{ij} = sim\big(e_{U_{i_B}}, e_{B_j}\big)$. Then the portion of the mean square error loss that $u_i$ contributes is $\Li = \frac{1}{N_I} \sumPMinusISquare$.

The size of $\Li$ is influenced by $\|Y_i\|$, but we can mitigate this issue by pre-normalizing $Y_i$ to form $\widehat{Y}_{i} = \frac{Y_{i} - \overline{Y_{i}}}{\|Y_{i}\|}$ and using $\widehat{Y}_{i}$ to compute the User-Normalized MSE loss $L_{i_N} = \frac{1}{N_I} \sumPMinusISquareNormalized$.

However, there is still a significant issue with this loss: although $\|Y_i\|$ does not affect the magnitude of $L_{i_N}$, $\|P_i\|$ does, so $L_{i_N}$ is very sensitive to outliers, especially ones that make the value of $\|P_i\|$ large.
Note that in a sparse matrix factorization setting each outlier user $u_o$ will not dramatically impact the optimization, since $\frac{\partial L_{{o}_N}}{\partial \theta}$ is only nonzero for $u_o$'s embedding vector and the embedding vectors of the items that $u_o$ interacted with. However, in a model like ImplicitCE each outlier user has a larger impact, since $\frac{\partial L_{{o}_N}}{\partial \theta}$ is potentially nonzero for all of weights of the $f(e_{U_A}; \theta)$ model as well as the embeddings of all the auxiliary and target domain items that $u_o$ interacted with.

Since we don't care about the magnitudes of the elements in $P_i$ and are only interested in their relative values, we can address this issue by normalizing $P_i$ as well to form $\widehat{P}_{i} = \frac{P_{i} - \overline{P_{i}}}{\|P_{i}\|}$. Then, our new per-user loss becomes:
$\widehat{L_{i}} = \frac{1}{N_I} \sum\limits_{}^j \big(\widehat{P}_{ij} - \widehat{Y}_{ij}\big)^{2} = 
\mathlarger{2}
\left(
\mathlarger{1} - \frac {
\mathlarger{\sum}\limits_{}^j \bigg(\widehat{P}_{ij} - \overline{\widehat{P}_{i}}\bigg) \bigg(\widehat{Y}_{ij} - \overline{\widehat{Y}_{i}}\bigg)
}{
\sqrt{\SumPhatminusPbar \SumIhatminusIbar}
}
\right)$
Note that this is equivalent to $2\left(
1 - corr(P_i, I_i)
\right)$, where $corr$ is the Pearson correlation coefficient estimator. By removing the constant term and averaging over all $N_U$ users, we form the Per-User Correlation Loss: $\Lcorr = 
\frac{1}{N_U}\sum\limits_{}^{i} \left(1 - corr(P_i, I_i)\right)$. By using $\Lcorr$ as the loss function we directly maximize the correlations between our model's predictions and each user's actual numbers of interactions.

\subsubsection{Convergence Rate Experiment} \label{PerUserCorrExperimentSection}

\begin{figure*}
\centering
\includegraphics[height=3in,  width=5in, keepaspectratio]{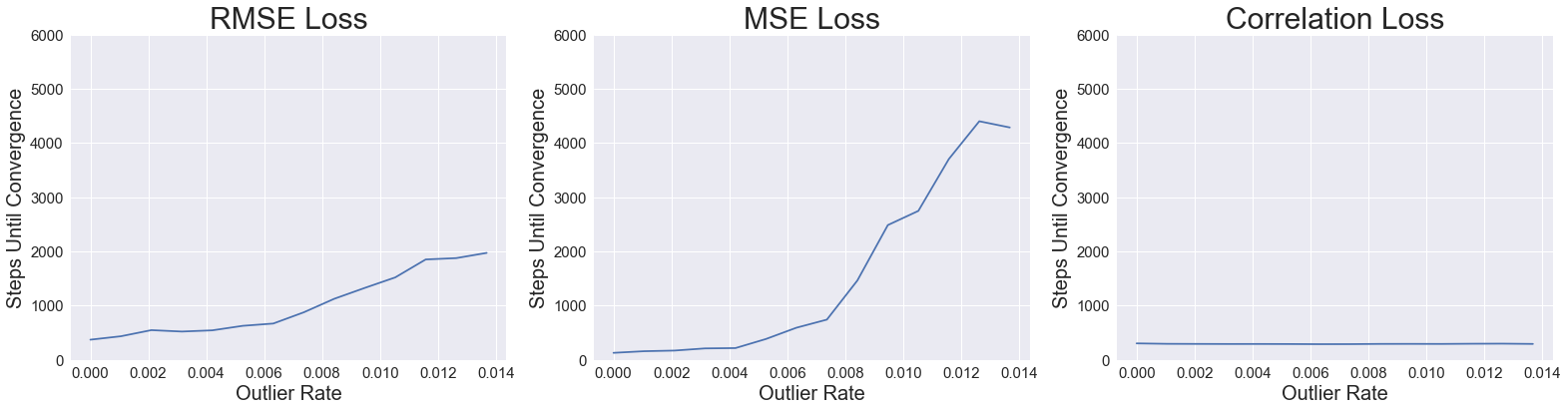}
\caption{Learning a linear function in the presence of outlier users. User-Normalized MSE and RMSE take longer to converge as the outlier users increase, but Per-User Correlation Loss does not.}
\label{fig:flat_steps_until_convergence}
\end{figure*}

To demonstrate the advantage that the Per-User Correlation Loss has over User-Normalized MSE, we perform a small experiment with simulated data. In order to illustrate that User-Normalized MSE's sensitivity to outliers is not simply an artifact of the squared term, we also include results over the User-Normalized RMSE loss: $\frac{1}{N_I} \sumSQRTPMinusISquareNormalized$

In this experiment, we use gradient descent to train an ordinary linear regression model to learn a mapping between simulated auxiliary and target domain item interaction data. We model users' auxiliary domain item interactions with a Multivariate Gaussian and we assign users' target domain item interactions ($Y_i$) to be a fixed linear function of their auxiliary domain item interactions.

For each loss function and outlier user rate $p$ we repeat the following process until convergence. \footnote{We define convergence as the loss function dipping below $10$ for User-Normalized RMSE, under $50$ for User-Normalized MSE and under $0.01$ for Per-User Correlation Loss (i.e. correlation > $0.99$).}
\begin{compactitem} \itemsep0em
\item Generate a "user" by drawing an auxiliary domain interaction vector from our Gaussian and computing the associated target domain interaction vector $Y_i$ with our fixed linear function.
\end{compactitem} \itemsep0em
\begin{compactitem} \itemsep0em
\item Generate a prediction $P_i$ for this "user" with our linear regression model and take a gradient descent step over all items $j$ towards minimizing the loss. 
\end{compactitem} \itemsep0em
\begin{compactitem} \itemsep0em
\item With probability $p$, repeat the above two steps with an "outlier user" that has a large number of auxiliary domain item interactions (and therefore a large $\|P_i\|$ since our model is linear) and random $Y_i$.
\end{compactitem} \itemsep0em

We find that as we increase the outlier user rate the User-Normalized MSE/RMSE models take longer to converge while the Per-User Correlation Loss's convergence rate remains unchanged.
(Figure ~\ref{fig:flat_steps_until_convergence}).

\subsubsection{Sample Correlation Update}

However, there is a serious problem with the Per-User Correlation Loss function that makes it infeasible to use with SGD over batches of (user, item) pairs when $N_I$ is large. In order to compute $\frac{\partial \Licorr}{\partial Pij}$ for a (user, item) pair $i,j$ we need to compute 
$\|P_{i}\| = \frac{1}{N_I}\sum\limits_{}^j P_{ij}^{2}$, which requires a sum with $N_I$ terms. To address this issue we apply the following simple algorithm, which we call Sample Correlation Update, or SCU:

\begin{compactenum} \itemsep0em 
\item Uniformly sample a small set of users $S_U$ with size $N_{S_U}$ and a small set of items $S_I$ with size $N_{S_I}$.

\item Compute $Pij$ for $i\in S_U,j\in S_I$, and the means $\overline{P_{iS_I}},\overline{Y_{iS_I}}$ over $i \in S_I$ to compute the following loss function: $\LcorrS = \bigUserSumS \perUserLossS$

\item Use the gradient of this loss $\nabla \LcorrS = \negativeBigUserSumS  \gradientPearsonLossS$, which only requires sums over $j\in S_I,i \in S_U$, to perform an update step.
\end{compactenum}

Since the error of the sample approximation of correlation and its gradient decrease exponentially as the size of the sample increases \cite{fisher1921}, we would expect that $L^{c_S}$ and $\nabla\LcorrS$ would quickly converge to $L^c$ and $\nabla\Lcorr$ as we increase $N_{I_S}$. This is exactly what we observe in Figure ~\ref{fig:correlation_prediction_error}: both $\|L^c - L^{c_S}\|$ and $\| \nabla\Lcorr - \nabla\LcorrS\|$ decrease exponentially as $N_{I_S}$ increases. We can now prove the following theorum about SCU:

\begin{figure*}
\centering
\includegraphics[height=3in,  width=5in, keepaspectratio]{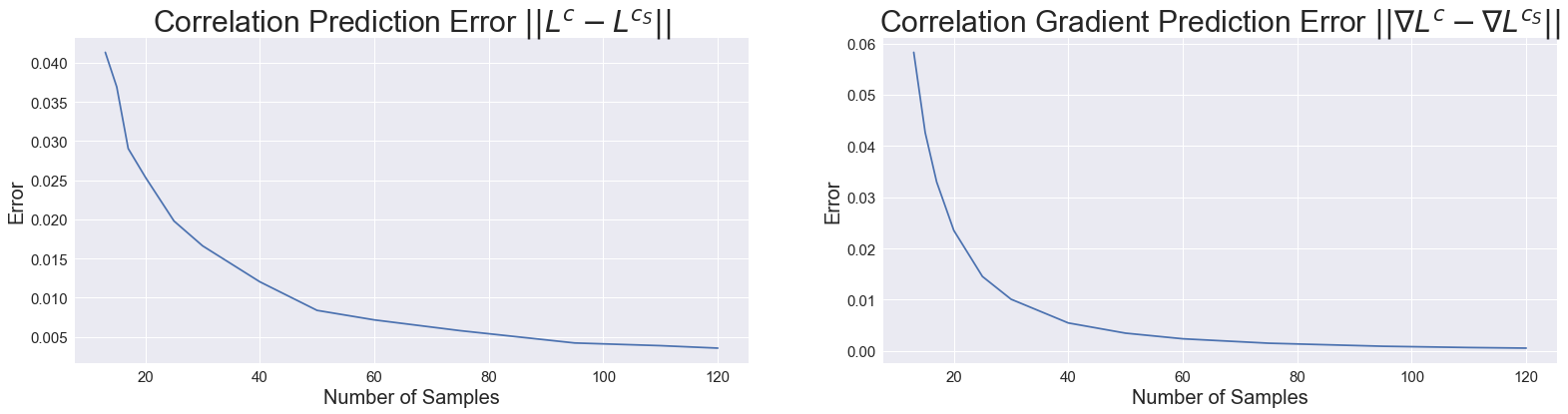}
\caption{If we generate random affinity vectors and predictions from a uniform distribution, we find that the square errors of both the sample approximation of correlation and its gradient decrease exponentially as the number of item samples increases.}
\label{fig:correlation_prediction_error}
\end{figure*}

\begin{lemma}
$\mathbb{E}_{S_U}\mathbb{E}_{S_I}[\nabla_j\LcorrS] = \nabla_j\Lcorr + \mathcal{O}(1/N_{I_S})$:
\end{lemma}

\begin{proof}
First, write $\mathbb{E}_{S_U}\mathbb{E}_{S_I}[\nabla_j\LcorrS$] as $\mathbb{E}_{S_U}\mathbb{E}_{S_I}\left[\nabla_j \smallUserSumS (1 - corr_{S_I}(P_i, I_i))\right]$. Since we can express $\mathbb{E}_{S_U}$ and $\mathbb{E}_{S_I}$ as sums, this is equivalent to $\nabla_j \mathbb{E}_{S_U}\left[\smallUserSumS \left(1 - \mathbb{E}_{S_I}\left[corr_{S_I}(P_i, I_i)\right]\right)\right]$

Now, let's note that sample correlation $corr_{S_I}(P_i, I_i)$ is not an unbiased estimator of population correlation $corr(P_i, I_i)$, but by \cite{fisher1915} we can write $\mathbb{E}_{S_I} [corr_{S_I}(P_i, I_i)]$ as:
\begin{align*}
corr(P_i, I_i) {-} \frac{corr(P_i, I_i) {-} corr(P_i, I_i)^3}{2N_{I_S}} {+} \mathcal{O}(1/N_{I_S}^2) &= \\
corr(P_i, I_i) {-} \mathcal{O}(1/N_{I_S}) 
\end{align*}  

This implies $\mathbb{E}_{S_U}\mathbb{E}_{S_I}[\nabla_j\LcorrS]$ is equivalent to:
\begin{align*}
\nabla_j \mathbb{E}_{S_U}\left[\smallUserSumS \Licorr + \mathcal{O}(1/N_{I_S}) \right] &= 
\\
\smallUserSumS \mathbb{E}_{S_U}\left[\nabla_j\Licorr\right] + \mathcal{O}(1/N_{I_S})
\end{align*}

Since $S_U$ is formed by uniformly sampling users, $\mathbb{E}_{S_U}\left[\nabla_j\Licorr\right] = \nabla_j\Lcorr$ and we can write:
$\mathbb{E}_{S_U}\mathbb{E}_{S_I}\nabla_j\LcorrS = \smallUserSumS \nabla_j\Lcorr+ \mathcal{O}(1/N_{I_S}) = \nabla_j\Lcorr+ \mathcal{O}(1/N_{I_S})$
\end{proof}

\section{Twitter and DBLP Experiments}

In order to evaluate ImplicitCE and SCU, we first compare their performance at generating recommendations for new users during their sign-up flow. In this experiment we use a large scale real world Twitter dataset. Although many collaborative filtering algorithms can utilize auxiliary domain information, relatively few are compatible with this problem.

In order to be suitable, an algorithm must be able to immediately predict target domain affinities for a new user who is not present at model-fitting time and has no interactions in the target domain. We've selected several of the most popular models that meet this criteria as baselines. To maintain consistency, we use $300$ element embeddings for each model (a common size for model comparisons \cite{word2vec}).

In order to further demonstrate the effectiveness of ImplicitCE and provide a fairer comparison we also evaluate the models on the open DBLP citation network dataset (aminer.org/citation).

\subsection{Models}

\begin{compactitem} \itemsep0em
\item \textit{Matrix Factorization with Fold-In}:
We construct a user-item affinity matrix where each row represents a training set user and each column represents an item from the auxiliary or target domains. We then factorize this matrix with either SVD, the ALS approach suggested in \cite{implicit}, or Collective Matrix Factorization \cite{collectivematrixfactorization} with fully shared latent factors and hyperoptimized matrix weights. In order to generate predictions for new users we apply the folding in techniques described in \cite{svdfoldin} and \cite{alsfoldin} to their auxiliary domain user-item affinity vectors. %
\end{compactitem} 

\begin{compactitem} \itemsep0em
\item \textit{CD-CCA}: 
We apply the techniques from \cite{ccacold} to generate low dimensional representations of users' auxiliary and target domain user-item affinity matrices, compute a mapping between them with Canonical Correlation Analysis, and reconstruct the target domain user-item affinity matrix from the low dimensional prediction.
\end{compactitem} 

\begin{compactitem} \itemsep0em
\item \textit{EMCDR}:
First, we generate low dimensional representations of the user's auxiliary and target domain user-item affinity matrices. Next, we use a neural network to learn a mapping from the auxiliary domain user embeddings to the target domain user embeddings. Finally, we use the similarities between the target domain user and item embeddings to approximate the target domain user-item affinities \cite{embeddingtoembedding}.
\end{compactitem} 

\begin{compactitem} \itemsep0em
\item \textit{ImplicitCE}:
We train ImplicitCE with SCU as well as mini-batch gradient descent with the User-Normalized MSE Loss and the Bayesian Personalized Ranking loss (see Section \ref{section:ImplicitCESection}). We train these models with each of three embedding similarity functions $sim$: Dot Similarity $uv$, Cosine Similarity $\frac{uv}{\|u\|\|v\|}$ and Euclidian Similarity $1 - \|u - v\|$.
\end{compactitem} 

For each model we ran a random search over the model hyperparameters to find the configuration that performs best on a validation set. We then evaluate that model on a holdout set. In order to compare models, we computed the averages and 95\% confidence intervals of the NDCG, ERR, Recall $@$ $10$ and Pearson Correlation metrics over all users in the holdout set.

\subsection{Twitter Experiment}
\begin{table*}[t]
\caption{Model performance on the Twitter and DBLP datasets by 95\% confidence intervals over the holdout set. The hyperparameters for all models were chosen by random search over a validation set.} 
\begin{subtable}[t]{1\textwidth}
\centering

\begin{tabular}{@{\extracolsep{4pt}}llccccccc}
\textbf{Twitter} & & Correlation & NDCG & ERR & 
Recall $@$ 10 \\ 
\toprule
\multirow{4}{*}{ImplicitCE}  
& Sample Corr Update
& \textbf{0.308 $\pm$ 0.002} 
& \textbf{0.533 $\pm$ 0.002} 
& \textbf{0.306 $\pm$ 0.002} 
& \textbf{0.891 $\pm$ 0.005} 
\tabularnewline
& MSE Loss
& 0.246 $\pm$ 0.002 
& 0.434 $\pm$ 0.003 
& 0.246 $\pm$ 0.002 
& 0.746 $\pm$ 0.006 
\tabularnewline
& BPR Loss
& 0.096 $\pm$ 0.001 
& 0.335 $\pm$ 0.002 
& 0.221 $\pm$ 0.002 
& 0.668 $\pm$ 0.005 
\tabularnewline
\midrule
\multirow{4}{*}{Baselines} 
& SVD
& 0.121 $\pm$ 0.009 
& 0.301 $\pm$ 0.014 
& 0.137 $\pm$ 0.012 
& 0.521 $\pm$ 0.018 
\tabularnewline
& Implicit ALS 
& 0.145 $\pm$ 0.001 
& 0.290 $\pm$ 0.007 
& 0.151 $\pm$ 0.002 
& 0.571 $\pm$ 0.028 
\tabularnewline
& Collective MF
& 0.111 $\pm$ 0.006 
& 0.280 $\pm$ 0.002 
& 0.140 $\pm$ 0.002 
& 0.498 $\pm$ 0.011 
\tabularnewline
& CD-CCA
& 0.189 $\pm$ 0.002 
& 0.312 $\pm$ 0.011 
& 0.212 $\pm$ 0.006 
& 0.511 $\pm$ 0.014 
\tabularnewline
& EMCDR
& 0.197 $\pm$ 0.009 
& 0.306 $\pm$ 0.027 
& 0.171 $\pm$ 0.006 
& 0.745 $\pm$ 0.006 
\tabularnewline
\bottomrule

\end{tabular}
\end{subtable}
\bigskip
\begin{subtable}[t]{1\textwidth}
\centering
\begin{tabular}{@{\extracolsep{4pt}}llccccccc}
\textbf{DBLP}  \\ 
\toprule
\multirow{4}{*}{ImplicitCE}  
& Sample Corr Update
& \textbf{0.577 $\pm$ 0.007} 
& \textbf{0.617 $\pm$ 0.005} 
& \textbf{0.592 $\pm$ 0.007} 
& \textbf{0.997 $\pm$ 0.002} 
\tabularnewline
& MSE Loss
& 0.501 $\pm$ 0.012 
& 0.444 $\pm$ 0.007 
& 0.532 $\pm$ 0.003 
& 0.956 $\pm$ 0.003 
\tabularnewline
& BPR Loss
& 0.401 $\pm$ 0.006 
& 0.471 $\pm$ 0.007 
& 0.558 $\pm$ 0.011 
& 0.965 $\pm$ 0.009 
\tabularnewline
\midrule
\multirow{4}{*}{Baselines} 
& SVD
& 0.250 $\pm$ 0.012 
& 0.491 $\pm$ 0.009 
& 0.581 $\pm$ 0.005 
& 0.908 $\pm$ 0.011 
\tabularnewline
& Implicit ALS 
& 0.235 $\pm$ 0.001 
& 0.466 $\pm$ 0.007 
& 0.532 $\pm$ 0.002 
& 0.923 $\pm$ 0.028 
\tabularnewline
& Collective MF
& 0.230 $\pm$ 0.002 
& 0.452 $\pm$ 0.002 
& 0.555 $\pm$ 0.001 
& 0.900 $\pm$ 0.009 
\tabularnewline
& CD-CCA
& 0.312 $\pm$ 0.009 
& 0.502 $\pm$ 0.015 
& 0.537 $\pm$ 0.003 
& 0.891 $\pm$ 0.003 
\tabularnewline
& EMCDR
& 0.301 $\pm$ 0.011 
& 0.494 $\pm$ 0.012 
& 0.550 $\pm$ 0.008 
& 0.938 $\pm$ 0.004 
\tabularnewline
\bottomrule

\end{tabular}
\end{subtable}
\label{table:CrossModelTable}
\end{table*}
\begin{table*}
\vspace{-5mm}
\centering
\caption{ImplicitCE performance on the Twitter dataset over model hyperparameters by 95\% confidence intervals over the holdout set. \label{table:CrossLayerTable}} 
\begin{tabular}{@{\extracolsep{4pt}}llccccccc}
& Correlation & NDCG & ERR & Recall $@$ 10 \\ 
\midrule

Linear $f(e_{U_A};\theta)$ 
& 0.259 $\pm$ 0.001 
& 0.489 $\pm$ 0.002 
& 0.302 $\pm$ 0.002 
& 0.881 $\pm$ 0.003 
\\
One Layer $f(e_{U_A};\theta)$
& 0.278 $\pm$ 0.003 
& 0.506 $\pm$ 0.004 
& 0.306 $\pm$ 0.002 
& 0.887 $\pm$ 0.021 
\\
Two Layer $f(e_{U_A};\theta)$
& 0.308 $\pm$ 0.002 
& \textbf{0.533 $\pm$ 0.002} 
& 0.306 $\pm$ 0.004 
& \textbf{0.891 $\pm$ 0.005} 
\\
Three Layer $f(e_{U_A};\theta)$
& \textbf{0.318 $\pm$ 0.001} 
& 0.529 $\pm$ 0.002 
& \textbf{0.307 $\pm$ 0.002} 
& 0.890 $\pm$ 0.003 
\\
\midrule
Cosine Similarity
& \textbf{0.308 $\pm$ 0.0018} 
& \textbf{0.533 $\pm$ 0.0024} 
& \textbf{0.306 $\pm$ 0.0022} 
& \textbf{0.891 $\pm$ 0.0046} 
\\
Dot Similarity
& 0.231 $\pm$ 0.0020 
& 0.396 $\pm$ 0.0016 
& 0.234 $\pm$ 0.0024 
& 0.722 $\pm$ 0.0064 
\\
Euclidian Similarity
& 0.228 $\pm$ 0.0018 
& 0.434 $\pm$ 0.0026 
& 0.256 $\pm$ 0.0024 
& 0.769 $\pm$ 0.0060 
\\
\bottomrule
\end{tabular}

\end{table*}
\begin{table*}[t]
\centering
\caption{Validation ROC-AUC for each of the topic prediction tasks by 95\% confidence intervals over the cross-validation folds. } 
\begin{tabular}{@{\extracolsep{4pt}}llccccccc}
& Sports & Music & Entertainment & Government \& Politics & News \\ 
\midrule

SVD 
& 0.730 $\pm$ 0.016 
& 0.568 $\pm$ 0.020 
& 0.624 $\pm$ 0.018 
& 0.618 $\pm$ 0.012 
& 0.623 $\pm$ 0.020 
\\
ALS 
& 0.739 $\pm$ 0.022 
& 0.589 $\pm$ 0.008 
& 0.626 $\pm$ 0.008 
& 0.650 $\pm$ 0.008 
& 0.622 $\pm$ 0.014 
\\
Autoencoder
& 0.602 $\pm$ 0.026 
& 0.575 $\pm$ 0.024 
& \textbf{0.675 $\pm$ 0.032} 
& 0.598 $\pm$ 0.050 
& 0.639 $\pm$ 0.052 
\\
EMCDR
& 0.651 $\pm$ 0.095 
& 0.550 $\pm$ 0.044 
& 0.511 $\pm$ 0.012 
& 0.502 $\pm$ 0.090 
& 0.601 $\pm$ 0.083 
\\
ImplicitCE
& \textbf{0.781 $\pm$ 0.012} 
& \textbf{0.696 $\pm$ 0.012} 
& 0.671 $\pm$ 0.012 
& \textbf{0.735 $\pm$ 0.012} 
& \textbf{0.726 $\pm$ 0.012} 
\\
\bottomrule
\end{tabular}
\label{table:TTTTable}
\end{table*}

On Twitter, users with large and active sets of followers are known as ``producers.'' Producers generate new content on a regular basis, and for each user it is important to recommend producers with whom they are likely to heavily interact (via likes, retweets, etc). When a new user registers for Twitter it is particularly important to recommend producers to them immediately after sign-up so they can start interacting with content that interests them. Since at this stage the user has not yet interacted with any producers, it is not possible to apply traditional recommendation techniques. However, before many users sign up for Twitter they interact with the platform in an indirect way by visiting web domains that have embedded Twitter content, such as embedded Tweets or Timelines. We refer to these domains as Twitter for Websites (TFW) domains. Since this embedded content is often closely related to the content that Twitter producers create, we can use past TFW domain interactions to predict future affinities for Twitter producers.

We evaluate our model on the task of predicting producer affinity (target domain) from observed TFW domain affinity (auxiliary domain).

In order to reduce noise and maximize the consistency between interactions and affinity we require at least 40 interactions with both TFW domains and Twitter producers. All in all, our dataset contains $359,066$ users, $95,352$ TFW domains and $829,131$ producers. We hold out $10,000$ users for each of the validation and holdout sets.

\subsection{DBLP Experiment}
The DBLP citation network contains information about articles published in academic venues over the past several decades. We use each authors’ co-author publication counts before $2013$ (auxiliary domain) to predict their post $2013$ conference publications (target domain). We consider each of the $25,210$ authors with at least 10 publications both before and after $2013$ to be a "user", each of the $507,516$ coauthors that these authors published with before $2013$ to be an auxiliary domain item, and each of the $3,070$ conferences that these authors published in after $2013$ to be a target domain item. We hold out $5,000$ users for each of the validation and holdout sets.

\subsection{Evaluation}

We find that on both datasets and over all metrics ImplicitCE trained with SCU significantly outperforms all of the baseline models (Table ~\ref{table:CrossModelTable}). Furthermore, we find that SCU significantly outperforms the BPR and MSE loss functions (Table ~\ref{table:CrossModelTable}). Among the baseline models, we find that the CD-CCA and EMCDR models significantly outperform the matrix factorization models on the Twitter dataset and slightly outperform them on the DBLP dataset. This makes sense, given that these models more directly solve the problem and are capable of modeling the auxiliary and target domain entities separately.

Our top performing ImplicitCE architecture is a two layer neural network with $1024$ units per layer, batch normalization and a $relu$ activation function. We trained the model with an Adam optimizer with a learning rate of $0.05$, a dropout rate of $0.3$, an $L_2$ weight penalty of $0.001$, and a cosine embedding similarity function. For SCU, we used $N_{S_I} = 1000$ and $N_{S_U} = 64$. We observed that although replacing the neural network with a linear model does reduce performance, the effect is not dramatic (Table ~\ref{table:CrossLayerTable}). 

\section{Twitter Topic Prediction Experiment}

The co-embeddings that ImplicitCE generates are powerful generic representations of user preferences in the target domain. This is especially useful when a user is new to the system that hosts the target domain items and the user's auxiliary domain interactions are the only user information that the system has access to. We can see this more clearly by stating the co-embedding property for dot product similarity in a different way: constructing a target domain co-embedding is equivalent to constructing user embeddings such that for each item $b_i$, the performance of a linear model trained on these embeddings to predict user affinity towards $b_i$ is maximized. This property suggests that these embeddings may also be strong low dimensional user representations for tasks that are similar to predicting target domain item affinity, such as categorizing users, serving advertisements or predicting user demographics. We evaluate this in the following task.
 
Some Twitter users have chosen to directly indicate which topics interest them. We can evaluate the ability of the user embeddings that ImplicitCE generates to serve as efficient representations of users' Twitter preferences by training a logistic regression model on them to predict these interest topic selections.

In this task we use a variety of methods to transform users' TFW Domain affinities into user embeddings and train logistic regression models on these embeddings to predict ``indicated''/``did not indicate'' for each of the ``Sports'', ``Music'', ``Entertainment'', ``Government \& Politics'', ``News'', and categories. We use a small dataset of $3000$ users and perform $20$-fold cross validation over these users. We quantify the model's performance with the mean and 95\% confidence interval of the cross validation ROC-AUC for each topic.

Since the topic prediction task requires the model to predict users' actions within the Twitter platform, it's possible that a method that generates embeddings that are finely tuned to reflect users' affinities within the Twitter platform is particularly useful for this task, especially since patterns of web domain affinity are significantly different from patterns of Twitter account affinity. For example, a particular web domain likely hosts content that appeals to a wider variety of viewpoints and interests than a particular Twitter account. Therefore, canonical low-dimensional representations of web domain content that are optimized for a web domain reconstruction objective may be suboptimal for a Twitter interest prediction task.

As a baseline we generate user embeddings by either factorizing the user-web domain interaction matrix with the SVD and ALS algorithms or compressing the user-web domain interaction matrix with an Autoencoder. We train all three models on the full web domain interaction dataset to generate embeddings of size $300$. These models attempt to generate representations that contain the maximum amount of information about a user's preferences for web domains, but do not attempt to represent user preferences for items within Twitter. We include the embeddings generated by the EMCDR model as an additional baseline.
We observe that for five out of the six topics, models trained on our embeddings outperform all baseline models (Table ~\ref{table:TTTTable}).

\section{Conclusion}
In this work, we present a novel algorithm and training strategy for recommending target domain items to users based on their historical interactions with auxiliary domain items. Our training strategy allows us to directly maximize the correlation between our model's predictions and user preferences, and our model's embedding structure allows us to generate recommendations in live production with minimal computational cost.

\def\bibfont{\footnotesize}
\bibliography{my_bibliography}
\bibliographystyle{aaai19}

\end{document}